\documentclass[11pt]{article}
\usepackage[final]{epsfig}
\usepackage[left=3cm,right=3cm, top=2.5cm,bottom=2.5cm,bindingoffset=0cm]{geometry}
\usepackage{graphics}
\usepackage{amsmath}
\usepackage{amsfonts}
\usepackage{latexsym}
\usepackage{amssymb, mathabx}
\usepackage{amsthm}
\usepackage{graphicx}
\usepackage{epstopdf}
\usepackage{multicol,multirow}
\usepackage{hyperref, enumitem}
\usepackage{mathtools}

\usepackage{mathrsfs}

 \usepackage{relsize}

\usepackage[nottoc]{tocbibind}

\usepackage{amsthm}

\usepackage[latin1]{inputenc}
\usepackage{tikz}
\usetikzlibrary{shapes,arrows,decorations.pathmorphing}
\usepackage{float}

\usepackage{blkarray}
\usepackage{xypic}

\makeatletter
\def\thmhead@plain#1#2#3{%
  \thmname{#1}\thmnumber{\@ifnotempty{#1}{ }\@upn{#2}}%
  \thmnote{ {\the\thm@notefont#3}}}
\let\thmhead\thmhead@plain

\makeatother

\newcounter{AppCounter}

\def\restrict#1{\raise-.5ex\hbox{\ensuremath|}_{#1}}

\newtheorem{lemma}{Lemma}[section]
\newtheorem{proposition}[lemma]{Proposition}

\newtheorem{remark-definition}[lemma]{Remark-Definition}
\newtheorem{theorem}[lemma]{Theorem}
\newtheorem{corollary}[lemma]{Corollary}

\newtheorem{conjecture}[lemma]{Conjecture}
\newtheorem{proposition-conjecture}[lemma]{Proposition-conjecture}

\newtheorem{question}[lemma]{Question}

\theoremstyle{definition}

\newtheorem{definition}[lemma]{Definition}
\newtheorem{remark}[lemma]{Remark}

\newcommand{\marginnote}[1]
{
}

\newcounter{cy}

\newcounter{bk}

\newcounter{dps}

\usepackage{lipsum}

\title{Higher-dimensional Euler fluids and Hasimoto transform: counterexamples and generalizations}

\author{ Boris Khesin\thanks{Department of Mathematics, University of Toronto, Toronto, ON M5S 2E4, Canada; e-mail:  \tt{khesin@math.toronto.edu}}
  ~and Cheng Yang\thanks{Department of Mathematics and Statistics, McMaster University, Hamilton, ON L8S 4K1, Canada, and the Fields Institute, Toronto, ON M5T 3J1, Canada; 
 e-mail: \tt{yangc74@math.mcmaster.ca}
  }}

\date{}

\begin{document}

\maketitle
\begin{abstract}
The binormal (or  vortex filament) equation provides the localized induction approximation of the 3D incompressible Euler
equation. We present explicit solutions of the binormal equation
in higher-dimensions that collapse in finite time. The local nature of this phenomenon suggests 
 the appearance of singularity in nearby vortex blob solutions of the Euler equation in 5D and higher.

Furthermore, the Hasimoto transform takes the binormal equation to the NLS and barotropic fluid equations. 
We show that in higher dimensions the existence of such a transform 
would imply the conservation of the Willmore energy in skew-mean-curvature flows and 
present  counterexamples  for vortex membranes based on products of spheres. 
These (counter)examples imply that there is no straightforward generalization to higher dimensions of the 1D Hasimoto transform. We derive its replacement,  the evolution equations
for the mean curvature and torsion form for membranes, thus generalizing the barotropic fluid and Da Rios equations. 
\end{abstract}

\tableofcontents

\section{Introduction} \label{intro}

The vortex filament equation describes the motion of a curve in $\mathbb R^3$ under the binormal flow: each point on the curve is moving in the binormal direction with a speed equal to the curvature at that point. This equation is a ``local"
approximation of the 3D Euler equation for vorticity supported on a curve. 
The membrane binormal (or skew-mean-curvature) flow is a natural higher-dimensional generalization 
of the 1D binormal flow \cite{Jer, Sh, Kh}:  instead of curves in $\mathbb R^3$, one traces the evolution 
of codimension 2 submanifolds in $\mathbb R^d$ (called vortex membranes), where the velocity 
of each point on the membrane is given by the  skew-mean-curvature vector. The latter 
is the mean curvature vector to the membrane rotated in the normal plane by $\pi/2$. 
The binormal equations are Hamiltonian  in all dimensions  with respect to the so-called Marsden-Weinstein 
symplectic structure and the Hamiltonian functional given by the length of the vortex filament or, more generally, the volume of the membrane \cite{HaVi, Kh}.

These equations in any dimension arise as an approximation of the incompressible Euler equation in which 
the vorticity is supported on a membrane and the evolution is governed by local interaction only (``LIA" - localized induction approximation). Below we present explicit solutions of the LIA for the Euler equation based on sphere products 
and prove that some of them exist for a finite time only  and then collapse, see Theorem \ref{prop:Sphere_product}. The simplest such case is the motion of a three-dimensional vortex  membrane $ \mathbb S^1\times\mathbb S^2$ in $\mathbb R^5$, and 
the local nature of this collapse hints to the singularity in the higher-dimensional Euler equations. While the singularity problem for the 3D Euler (and Navier-Stokes) equations is well known and wide open, it is equally open in dimensions $n>3$. Note that the incompressibility condition is seemingly less restrictive in higher dimensions, and hence the incompressible Euler equation should behave somewhat similar to the Burgers or compressible Euler equations, where the emergence of shock waves, and hence no long-time existence, is well known.
In spite of this similarity, to the best of our knowledge, there are yet no explicit results about  emergence of singularity in higher-dimensional incompressible Euler equations. Hopefully, the sphere product example of the binormal motion of
$ \mathbb S^1\times\mathbb S^2\subset \mathbb R^5$ could shed some light on the finite-time existence of
a smooth solution of the Euler equation tracing this motion.

\smallskip

The second goal of this paper is to give a counterexample to the existence of a simple form of the Hasimoto transform 
in higher dimensions, while to present a  replacement of the corresponding equations.
Namely,  Da Rios in 1906 brought in the idea of LIA of the Euler equation to study the vortex dynamics and 
derived the evolution equations for the  curvature $\kappa$ 
and torsion $\tau$ of a curve moving according to the binormal flow. His work 
is known today mostly thanks to his advisor Levi-Civita, who promoted and extended it (see \cite{Ricc} for the historical survey).
The LIA method and Da Rios equations were reconsidered in the 1960s, see \cite{Betc}. 
These equations  appear in several other contexts, e.g., in the study of 
 one-dimensional classical spin systems \cite{ArKh, Cal}. 

More importantly, in 1972 Hasimoto discovered a transformation yielding a complex-valued wave function $\psi=\kappa \exp(i\int\tau)$ from the pair of real functions $(\kappa,\tau)$ such that 
this wave function $\psi$ satisfies the nonlinear Schr\"{o}dinger equation. Furthermore, by considering the evolution of the density  
$\rho=\kappa^2$ and the velocity be $v=2\tau$ one obtains the  equations of barotropic-type (quantum) 1D fluids, see Figure \ref{fig:diag1} for the relations between these equations.

A natural question is whether the higher-dimensional binormal flow possesses similar relations to
 Schr\"{o}dinger- and barotropic-type equations, as well as what are its implications for the Euler hydrodynamics. 
Finding a higher-dimensional version of the Hasimoto transform 
was a folklore problem for quite a while, see e.g. \cite{Jer,  Kh, SaWa, Sh, SongSun}.
It is showed in \cite{Song} that the Gauss map of the SMCF satisfies a Schr\"{o}dinger flow equation. 
It was observed in \cite{KhMM} that to construct a higher dimensional generalization of the Hasimoto transformation one needs
  to prove a  conservation law for the Willmore energy. Namely, the conjectural invariance of the Willmore energy would imply the simpler of the two barotropic fluid equations, the continuity equation, which is a necessary condition for the existence of a  Hasimoto transformation.

In the present paper, we give a counterexample to the energy invariance conjecture by describing explicitly the motion 
of Clifford tori under the skew-mean-curvature flow and show that their Willmore energy is not conserved, see Proposition \ref{prop:filament_Willmore}. 
Essentially, these counterexamples imply that there is no straightforward generalization of the Hasimoto transform 
to relate the binormal and barotropic (and hence Schr\"{o}dinger) equations in higher dimensions, and if it exists, it must be necessarily complicated. 

\smallskip

Finally, we  introduce a natural generalization of the torsion for codimension 2 membranes (Section \ref{sec:torsion})
and derive the evolution equations
for the mean curvature and the torsion form, thus replacing equations of a barotropic-type fluid in the Hasimoto transform and
generalizing the Da Rios equations, Theorems  \ref{eq:source}, \ref{thm: conti_eq}.
These counterexamples emphasize the difference between the 1D and higher-dimensional skew-mean-curvature flows and 
might be particularly useful to prove the vortex filament conjecture for membranes, cf. \cite{Jer}. 


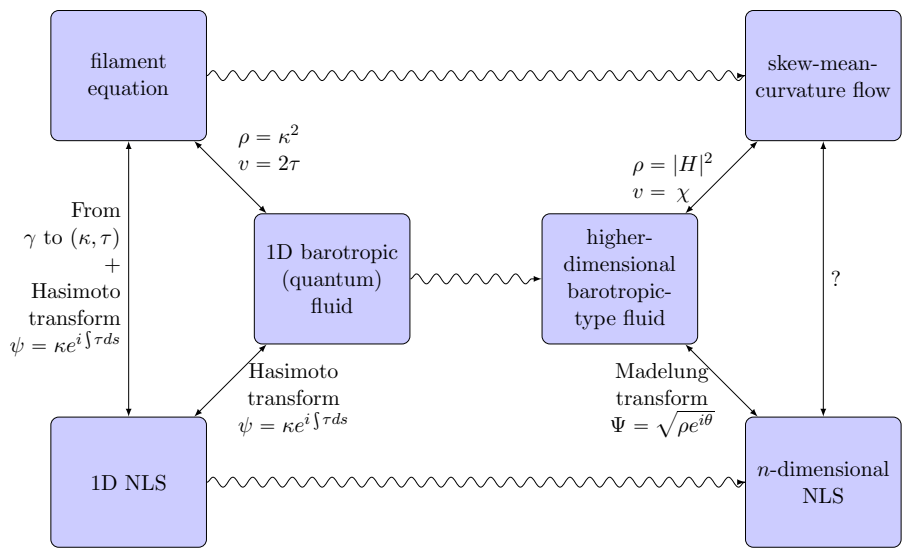
\begin{figure}[H]
\tikzstyle{block} = [rectangle, draw, fill=blue!20,
    text width=6.5em, text centered, rounded corners, minimum height=6em]
\tikzstyle{line} = [draw, -latex']
\centering
\resizebox{0.75\width}{0.75\height}{
\begin{tikzpicture}[node distance = 5.1cm,  auto]
    \node [block] (1DQF) {1D barotropic\\ (quantum) fluid};
    \node [block, above left of=1DQF] (filament) {filament equation};
    \node [block, below left of=1DQF] (NLS) {1D NLS};
    \node [block, right of=1DQF] (nDQF) {higher-dimensional barotropic-type fluid};
    \node [block, above right of=nDQF] (binormal) {skew-mean-curvature flow};
    \node [block, below right of=nDQF] (nNLS) {$n$-dimensional NLS};
    \draw[<->, latex-latex](filament) -- node [align=right,left] {From \\$\gamma$ to $(\kappa,\tau)$\\+\\Hasimoto\\transform\\$\psi=\kappa e^{i\int\tau ds}$} (NLS);
    \draw[<->, latex-latex](filament) -- node [align=left] {$\rho=\kappa^2$\\$v=2\tau$} (1DQF);
    \draw[<->, latex-latex](NLS) -- node [align=center,right] {\\\\ Hasimoto\\transform\\$\psi=\kappa e^{i\int\tau ds}$} (1DQF);
    \draw[->, -latex,  decorate, decoration=snake](1DQF) --  (nDQF);
    \draw[->, -latex, decorate, decoration=snake](filament) --  (binormal);
    \draw[->, -latex, decorate, decoration=snake](NLS) --  (nNLS);
    \draw[<->, latex-latex](binormal) -- node [right] {?} (nNLS);
    \draw[<->, latex-latex](binormal) -- node [align=left,left] {$\rho=|H|^2$\\$v=\,\chi$} (nDQF);
    \draw[<->, latex-latex](nNLS) -- node [align=center,left] {\\\\Madelung\\transform\\$\Psi=\sqrt{\rho e^{i\theta}}$} (nDQF);
\end{tikzpicture}
}
\caption{Diagram of relations between equations in 1D and in higher dimensions} \label{fig:diag1}
\end{figure}

\medskip

{\bf Acknowledgments.} We are indebted to R.~Jerrard and B.~Shashikanth for many fruitful discussions.
B.K. was partially supported by an NSERC research grant. A part of this work was done
while C.Y. was visiting  the Fields Institute in Toronto and the Instituto de Ciencias Matem\'{a}ticas (ICMAT) in Madrid. C.Y. is grateful for their supports and kind hospitality.


\medskip

\section{Skew-mean-curvature flows}\label{sec:smcf}

\subsection{The vortex filament equation}\label{subsec:filament}

Consider the space of (nonparametrized) knots $\mathfrak K$ in $\mathbb R^3$, which is the set of images of all smooth embeddings $\gamma:S^1\rightarrow\mathbb R^3$. 

\begin{definition}
The {\it vortex filament  equation} is $\partial_t \gamma=\gamma'\times\gamma'',$
where $\gamma':=\partial \gamma/\partial s$ with respect to the arc-length parameter $s$ 
of the curve $\gamma$. Alternatively, the filament equation can be rewritten in the
{\it binormal form} as 
\begin{equation}\label{eq:filament}
\partial_t\gamma=\kappa\,\bf{b}\,, 
\end{equation}
where, respectively, $\kappa$ is the curvature  and $\bf{b}={\bf t}\times {\bf n}$ is the binormal vector, the cross-product of the tangent and normal unit vectors,  at the corresponding point of the curve $\gamma$.

It is known that the binormal equation is Hamiltonian with respect to the so-called Marsden-Weinstein symplectic structure  on the space of knots $\mathfrak K$, the corresponding Hamiltonian function is the length functional $L(\gamma)=\int_{\gamma}|\gamma'(s)|\,ds$ of the curve.
\end{definition}

\begin{definition}\label{def:symp}
Let $\gamma\in\mathfrak K$ be an oriented space curve in $\mathbb R^3$, then the {\it Marsden-Weinstein symplectic structure} $\omega^{MV}$ on the space $\mathfrak K$ is given by
\begin{equation}\label{eq:MW}
\omega^{MV}(\gamma)(u,v)=\int_{\gamma}i_ui_v\mu=\int_{\gamma}\mu(u,v,\gamma')\,ds,
\end{equation}
where $u$ and $v$ are two vector fields attached to $\gamma$, and $\mu$ is the volume form in $\mathbb R^3$.
\end{definition}


The vortex filament equation also serves as an approximation for the 3D incompressible Euler equation for the vorticity
confined to the curve $\gamma$ (hence, the name), where only local interaction is taken into account \cite{ArKh, Cal}, cf. Section \ref{sect:vortex}.


\medskip

\subsection{Higher-dimensional  binormal flows}

The higher-dimensional generalization of the 1D binormal flow is also called the skew-mean-curvature flow, and it is defined as follows:

\begin{definition}\label{def:smcf}
Let $\Sigma^n\subset\mathbb R^{n+2}$ be a codimension  2 membrane (i.e., a compact oriented submanifold of codimension 2 in the Euclidean space $\mathbb R^{n+2}$), the {\it skew-mean-curvature (or, binormal) flow} is described by the equation:
\begin{equation}\label{eq:smcf}
\partial_t p =-J(H(p)),
\end{equation}
where $p\in \Sigma$, $H(p)$ is the mean curvature vector to $\Sigma$ at the point $p$, $J$ is the operator of positive $\pi/2$ rotation in the two-dimensional normal space $N_p\Sigma$ to $\Sigma$ at $p$.

The skew-mean-curvature flow  (\ref{eq:smcf}) is a natural generalization of the binormal equation \cite{Jer}:
in dimension $n = 1$ the mean curvature vector of a curve $\gamma$ at a point is $H=\kappa\,\bf{n}$, where $\kappa$ is the curvature of the curve $\gamma$ at that point, hence the skew-mean-curvature flow becomes the binormal equation
\eqref{eq:filament}:
$\partial_t\gamma=-J(\kappa\,\bf{n}) = \kappa\,\bf{b}$. It was studied for codimension 2 vortex membranes in $\mathbb R^4$ in \cite{Sh} and in any dimension in \cite{Kh}.
\end{definition}

It turns out that on the  infinite-dimensional space  $\mathfrak M$ of codimension 2 membranes, one can also define the Marsden-Weinstein symplectic structure in a similar way:

\begin{definition}\label{def:symp2}
The Marsden-Weinstein symplectic structure $\omega^{MV}$ on the space $\mathfrak M$  of codimension 2 membranes is
\begin{equation}\label{eq:MW2}
\omega^{MV}(\Sigma)(u,v)=\int_{\Sigma}i_ui_v\mu,
\end{equation}
where $u$ and $v$ are two vector fields attached to the membrane $\Sigma\in\mathfrak M$, 
and $\mu$ is the volume form in $\mathbb R^{n+2}$.
\end{definition}

Define the Hamiltonian functional ${\rm vol}(\Sigma)$ on the  space $\mathfrak M$ which associates  the $n$-dimensional volume to a compact $n$-dimensional membrane $\Sigma^n\subset \mathbb R^{n+2}$.

\begin{proposition}
The skew-mean-curvature flow \eqref{eq:smcf} 
is the Hamiltonian flow on the membrane space $\mathfrak M$ equipped with the
Marsden-Weinstein structure and with the Hamiltonian given by the volume functional  ${\rm vol}$.
\end{proposition}

\begin{proof}
In a nutshell, the Marsden-Weinstein  symplectic structure is the averaging of the symplectic structures
in all 2-dimensional normal planes $N_p\Sigma$ to $\Sigma$, hence the skew-gradient for any functional on submanifold 
$\Sigma$ is obtained from its gradient field attached at $\Sigma\subset \mathbb R^{n+2}$ by applying the fiberwise 
$\pi/2$-rotation operator $J$ in $N_p\Sigma$.
On the other hand, the fact that minus the mean curvature vector field $-H$ is the gradient
for the volume functional ${\rm vol}(\Sigma)$ is well-known, see e.g. \cite{KuSc, Kh}. Hence the Hamiltonian field on  $\mathfrak M$
for the Hamiltonian  functional ${\rm vol}(\Sigma)$ is given by $-JH(p)$ at any point $p\in \Sigma$.
\end{proof}


\medskip


\subsection{Collapse in binormal flows of sphere products}\label{subsec:Clifford}

Binormal flows are localized approximations of the Euler equation for an incompressible 
fluid filling $\mathbb R^{n+2}$ whose vorticity is supported on the membrane $\Sigma^n$, see  \cite{Kh, Sh} and the 
next section. This is why their short/long-time existence results could shed some light on the motion of fluid flows themselves.
It turns out that the following family of membrane motions is of particular interest.

\begin{theorem}\label{prop:Sphere_product}
Let $F:\Sigma=\mathbb S^m(a)\times\mathbb S^l(b)\hookrightarrow\mathbb R^{m+1}\times\mathbb R^{l+1}=\mathbb R^{m+l+2}$
be the product of two spheres of radiuses $a$ and $b$. Then the evolution  $F_t$ of this surface $\Sigma$
in the binormal flow is the product of spheres 
$F_t(\Sigma)=\mathbb S^m(a(t))\times \mathbb S^l(b(t))$ at any $t$ with radiuses changing monotonically according to the ODE system:
\begin{equation}\label{eq:CliffordODE2}
\left\{\begin{array}{rcl}
\dot a&=&-l/b,\\
\dot b&=&+m/a.
\end{array}\right.
\end{equation}
For $0<m<l$ the corresponding solution $F_t$ exists only for finite time and collapses 
at $t=a(0)b(0)/(l-m)$. 
\end{theorem}

\begin{corollary}\label{cor:2}
In the general case of sphere products $\Sigma=\mathbb S^m(a)\times\mathbb S^l(b)$ the radiuses of $F_t(\Sigma)$
change as follows: $a(t)=ae^{-lt/(ab)}$ and $b(t)=be^{mt/(ab)}$ for $m=l$ and
$$
a(t)=a^{m/(m-l)}\left(a-(l-m)b^{-1}t\right)^{l/({l-m})}\;\text{and}\;\; b(t)=b^{l/(l-m)}\left(b+({m-l})a^{-1}t\right)^{m/(m-l)},
$$
for  $m\neq l$ and initial conditions $a(0)=a$, $b(0)=b$. 
\end{corollary}

\begin{remark}
The simplest case satisfying  the collapse condition $0<m<l$  is $m=1, l=2$ for  $ \mathbb S^1(a)\times\mathbb S^2(b)\subset \mathbb R^5$.  Since the skew-mean-curvature flow is the localized induction approximation of the Euler equation, this explicit solution might be useful to study the Euler singularity problem in higher dimensions.
Note also that the odd-dimensional Euler equation has fewer invariants (generalized helicities) than the even-dimensional one
(generalized enstrophies), see \cite{ArKh}. The existence of many invariants helps control solutions, so it is indicative that 
the first example with a finite life-span occurs in the odd 5D.
\end{remark}

\begin{proof}
For a point $q=(q_1,q_2)\in \mathbb S^m(a)\times\mathbb S^l(b)\hookrightarrow\mathbb R^{m+1}\times\mathbb R^{l+1}$,
let ${\bf n}_1$ and ${\bf n}_2$ be the outer unit normal vectors to the corresponding spheres at the points $q_1$ and $q_2$ respectively.
Then the mean curvature vectors of $\mathbb S^m(a)$ and $\mathbb S^l(b)$ as hypersurfaces in $\mathbb R^{m+1}$ and $\mathbb R^{l+1}$ are $-\frac 1a {\bf n}_1$ and $-\frac 1b {\bf n}_2$ respectively. 
Therefore the total mean curvature vector $H$
of $F:\mathbb \mathbb S^m(a)\times\mathbb S^l(b)\rightarrow\mathbb R^{m+l+2}$ is a (normalized) contribution of $m$ vectors 
$-\frac 1a {\bf n}_1$ coming from $\mathbb S^m(a)$ and $l$ vectors $-\frac 1b {\bf n}_2$ coming from $\mathbb S^l(b)$. Thus the  mean curvature of the sphere product is the vector
$H=-\frac ma {\bf n}_1-\frac lb {\bf n}_2$ (divided by the total dimension $m+l$ of the product, which we omit),
and the skew-mean-curvature vector is $-JH=-\frac lb {\bf n}_1+\frac ma {\bf n}_2$.

This implies that for the skew-mean-curvature flow $\partial_t q = -JH(q)$
given by the above linear combination of the normals on the product of spheres,   $\Sigma_t$ remains the 
product of two spheres $S^m(a(t))\times S^l(b(t))$
for all times, where one of the spheres is shrinking, while the other is expanding. 

The explicit form of the $-JH$ vectors implies  the system of ODEs \eqref{eq:CliffordODE2} on the evolution of radiuses. Rewriting this as one first order ODE one can solve this explicitly, as in Corollary~\ref{cor:2}.
The system (\ref{eq:CliffordODE2}) is Hamiltonian on the $(a,b)$-plane with the Hamiltonian function given by $\mathcal H(a,b):=\ln(a^mb^l)$, which is the logarithm of the volume of the product of two spheres: ${\rm vol}(\Sigma)=C \,a^mb^l$.
(Note that the invariance of  this Hamiltonian is consistent with conservation of the volume 
of $\Sigma$, as the latter is the Hamiltonian  of the skew-mean-curvature flow.) 
\end{proof}


\subsection{Higher-dimensional Euler equation in the vortex form}\label{sect:vortex}

Explicit solutions of the binormal (or LIA for the Euler) equation based on sphere products discussed above
could shed some light on the singularity problem for the higher-dimensional Euler equation, as the skew-mean-curvature flow is an approximation of the Euler equation for vorticity supported on a membrane, cf. \cite{JerSm,  MaPu}. To the best of our knowledge, it is the first example of an explicit solution of the LIA existing for finite time, and  
the collapse or long-time existence of solutions of the 
binormal equation is suggestive for the corresponding properties of the hydrodynamical Euler solutions.

\smallskip

Recall that the classical Euler equation for an inviscid incompressible fluid in $\mathbb R^d$ describes 
an evolution of a divergence-free fluid velocity field $v(t,x)$: 
\begin{equation}\label{ideal}
\partial_t v+(v, \nabla) v=-\nabla p\,,
\end{equation}
where a pressure function $p$ is defined uniquely modulo an additive constant by decaying conditions at infinity and the constraint  ${\rm div}\, v=0$.

The binormal equation \eqref{eq:smcf} arises from the Euler equation as its localized induction approximation. Namely,
define the vorticity 2-form $\xi=dv^\flat$ for the 1-form $v^\flat$ related to the  divergence-free vector field $v$ by means of the Euclidean metric in $\mathbb R^n$. The vorticity form of the Euler equation is $\partial_t \xi=-L_v \xi,$
which means that the vorticity 2-form $\xi$ is transported by the flow.
The frozenness of the vorticity 2-form allows one to define various invariants of the  hydrodynamical Euler equation.

\begin{remark}
For $d=2$ and singular vorticity $\xi$, supported on a set of points in the plane, $\xi=\sum^N_{j=1} \Gamma_j\,\delta_{z_j}$, where
$ z_j\in\mathbb C\simeq \mathbb R^2$ are coordinates of the $j$th point vortex, the  evolution of point vortices according to the Euler equation is described by the Kirchhoff system
$$
\Gamma_j \dot z_j
=-J\frac{\partial  \mathcal H}{\partial z_j},\qquad 
1\le j\le N
$$
in $\mathbb C^N\simeq \mathbb R^{2N}$ for the  Hamiltonian function $\mathcal H=-\frac{1}{4\pi}\sum^N_{j<k}\Gamma_j\Gamma_k \, 
\ln | z_j- z_k|^{2}\,.$
\end{remark}

More generally, assume that the vorticity 2-form $\xi$ is a singular $\delta$-type form supported 
on a membrane $\Sigma$: $\xi=\delta_\Sigma$.
Then a (co-closed) 1-form $v^\flat=d^{-1}\delta_\Sigma$ (and hence the divergence-free vector field $v$) 
can be reconstructed by means of a Biot-Savart-type integral formula from the vorticity $\xi$. 
Finally, by keeping only local terms in the expression for the field-potential $v$ 
and rescaling the time variable in the Euler equation $\partial_t \xi=-L_v \xi$, 
one arrives at the binormal equation \eqref{eq:smcf} for the evolution of the vorticity  support $\Sigma$, see details in \cite{Kh}.

\medskip

\begin{remark}
There is yet another relation of the Euler and binormal equations in the case of sphere product membranes.
By assuming both the velocity $v$ and the pressure $p$ in equation \eqref{ideal}
to be functions of the 
distances $(x,y)$ to the origin: $x=|X|, y=|Y|$ for $X\in \mathbb R^{m+1}, \, Y\in \mathbb R^{l+1}$, 
one arrives at a version of the 2D Euler equation \eqref{ideal} supplemented by the adjusted incompressibility condition: 
${\rm div}\, (b(x,y)\cdot v)=0$, where $b(x,y):=x^m y^l$.  This  equation for a smooth function $b(x,y)$
in a bounded domain (also called the lake equation) was  studied  in \cite{DeSch}: the function $b(x,y)$ can be understood as  the lake's depth in a 
model of the vertically averaged horizontal velocity.

Then the examples of motion for the products of spheres correspond to singular vorticity $\xi=\delta_\Sigma$
for $\Sigma=\mathbb S^m(a)\times\mathbb S^l(b)\subset\mathbb R^{m+1}\times\mathbb R^{l+1}=\mathbb R^{m+l+2}$ for the Euler equation in $\mathbb R^{m+l+2}$. It reduces to the motion of a point vortex 
$\delta_{(a,b)}$ for $(a,b)\in \mathbb R^2_+$ for the corresponding lake equation.\footnote{We are grateful to R.~Jerrard for this remark.}   
Hence Theorem \ref{prop:Sphere_product} provides explicit solutions of point-vortex type, both existing forever or collapsing in finite time, depending on the membrane structure and dimension.
\end{remark}


\section{(Non)invariance of the Willmore energy and (non)existence of the Hasimoto transform}\label{sec:willmore}
\subsection{Motivation:  Hasimoto and Madelung}\label{subsec:Hasimoto}
It turns out that the example of vortex sphere products also delivers a counterexample for the existence of a simple analogue of the Hasimoto transform. To describe this counterexample we start with outlining  three different avatars of the skew-mean-curvature flows and the related  conjecture on the  Willmore energy conservation.
\smallskip

\begin{definition}[\cite{Hasi}]\label{def:hasi}
Given a parametrized curve $\gamma: \mathbb R\to \mathbb R^3$ with
curvature $\kappa$ and torsion $\tau$, the  {\it Hasimoto transformation} assigns 
the wave function $\psi:\mathbb R\to \mathbb C$ according to the formula
\begin{equation}\label{eq:Hasi}
(k(s),\tau(s)) \mapsto \psi(s) = \kappa(s)e^{i\int_{s_0}^s \tau(x)\,dx},
\end{equation}
where $s_0$ is some fixed point on the curve. (The ambiguity in the choice of $s_0$ defines the wave function $\psi$ up to a phase.)
\end{definition}
This Hasimoto map takes the vortex filament equation (\ref{eq:filament}) to the 1D nonlinear Schr\"odinger  (NLS) equation:
\begin{equation}\label{eq:NLS}
i\partial_t{\psi}+\psi''+\frac 12 |\psi|^2\psi=0
\end{equation}
for $\psi(\cdot, t):\mathbb R\to \mathbb C$, see e.g. \cite{Cal}. 

On the other hand,  considering separately the curvature $\kappa(\cdot, t)$ and torsion $\tau(\cdot, t)$ of the 
curve $\gamma(\cdot, t)\in\mathbb R^3$ moving by the binormal flow, the evolution of  $\kappa$ and $\tau$ satisfies 
the following system of {\it Da Rios' equations} \cite{DaRi}:
\begin{equation}\label{eq:DaRios}
\left\{\begin{array}{l}
\partial_t{\kappa}+2\kappa'\tau+\kappa\tau'=0,\\
\partial_t{\tau}+2\tau'\tau-\left(\frac{\kappa^2}{2}+\frac{\kappa''}{\kappa}\right)'=0.
\end{array}\right.
\end{equation}
By introducing the density $\rho=\kappa^2$ and the velocity $v=2\tau$, the Da Rios equations turn into the following system of compressible fluid equations:
\begin{equation}\label{eq:Compr}
\left\{\begin{array}{l}
\partial_t{\rho}+\text{div}(\rho v)=0,\\
\partial_t{v}+vv'+\left(-\rho-2\frac{\sqrt{\rho}''}{\sqrt{\rho}}\right)'=0.
\end{array}\right.
\end{equation}

What part of the above can be generalized to higher dimensions?
It turns out that long before the discovery of the Hasimoto transform, 
Madelung \cite{Made} gave a hydrodynamical formulation of
the Schrodinger equation in 1927, which is called the Madelung transform. 

\begin{definition}\label{def:Madelung}
Let $\rho$ and $\theta$ be real-valued functions on an $n$-dimensional manifold $M$ with
$\rho>0$. The {\it Madelung transform} is the mapping $\Phi :(\rho,\theta)\mapsto \psi$
defined by
\begin{equation}
\psi=\sqrt{\rho e^{i\theta}}.
\end{equation}
\end{definition}

The Madelung transform maps the system of equations for a barotropic-type fluid to the Schr\"{o}dinger equation. More specifically, let $(\rho,\theta)$ satisfy the following  barotropic-type fluid equations:
\begin{equation}\label{eq:barotropic2} 
\left\{ 
\begin{aligned} 
&\partial_t\rho +\text{div}(\rho v) = 0, \\
&\partial_t v + (v,\nabla) v + \nabla\Big(2V - 2f(\rho) - \frac{2 \Delta\sqrt{\rho}}{\sqrt{\rho}} \Big) = 0 
\end{aligned} \right.
\end{equation} 
with potential velocity field $v=\nabla \theta$, and functions $V\colon M\to \mathbb R$ and $f\colon (0,\infty) \to \mathbb R$. 
Then the complex-valued function $\psi(\cdot, t):=\sqrt{\rho e^{i\theta}}:M\to\mathbb C$ obtained by the Madelung transform satisfies the nonlinear Schr\"{o}dinger equation
\begin{equation}\label{eq:schrodinger} 
 	\mathrm{i}\partial_t\psi = - \Delta\psi +  V\psi - f(|\psi|^2)\psi.
\end{equation} 
In the 1D case for $V=0$ and $f(z)=z/2$ this gives the equivalence of the NLS \eqref{eq:NLS} and the compressible
fluid \eqref{eq:Compr}.

One can see that the one-dimensional Madelung transform, being interpreted in terms of the curvature and torsion of a curve, 
reduces to the Hasimoto transform  \cite{KhMM}. It is challenging, however, to fit the membrane geometry 
into this framework, and a search for a proper generalization of the Hasimoto map to different manifolds and to higher dimensions has been on  for some time, cf. e.g. \cite{Jer, Mol, SaWa, SongSun}.

The main question is whether there exists an analogue of the Hasimoto map which can send 
the binormal equation (\ref{eq:smcf}) to an NLS-type equation for any dimension $n$ \cite{KhMM},
or, thanks to the Madelung transform identifying  the NLS and the barotropic equations, 
one is searching for a relation 
of the binormal equation (\ref{eq:smcf}) and barotropic-type fluid equations \eqref{eq:barotropic2} in arbitrary dimension $n$.

In view of the binormal evolution \eqref{eq:smcf} and continuity equations \eqref{eq:DaRios}-\eqref{eq:Compr}, 
the square of the mean curvature vector $|H|^2$ is regarded as a natural
analogue of the density $\rho$ (recall that in 1D we set $\kappa^2=\rho$). 
Therefore an analogue of the total mass of the fluid  is the Willmore energy:

\begin{definition}\label{def:Willmore}
For an immersed submanifold $F:\Sigma^k\rightarrow\mathbb R^d$, its  {\it Willmore energy} is defined as
\begin{equation}\label{eq:willmore}
\mathcal W(F)=\int_{\Sigma}| H(F(q))|^2\;d\text{vol}_{g}=\int_{F(\Sigma)}| H(p)|^2\;d\text{vol}_{g_e},
\end{equation}
where $g=F^*g_e$ denotes the pull-back metric of the Euclidean metric $g_e$ on $\mathbb R^d$ and $H$ is the mean curvature vector at point $p=F(q)$ on the submanifold $F(\Sigma)\subset\mathbb R^d$.
\end{definition}

Assuming the existence of relation between the skew-mean-curvature flow and a barotropic fluid, one arrives at the following conjecture:

\begin{conjecture}\label{conj:Willmore}{\normalfont{\cite{KhMM}}}
For a codimension 2 submanifold $F_t:\Sigma^n\rightarrow \mathbb R^{n+2}$ moving by the skew-mean-curvature flow 
$\partial_t q=-JH(q)$  for $q\in\Sigma$ the following equivalent properties hold:

{\rm i)} its Willmore energy $\mathcal{W}(F_t)$ is invariant, 

{\rm ii)}
 its square mean curvature $\rho=|H|^2$ evolves according to the continuity equation 
 $$
 \partial_t \rho+{\rm div}(\rho v)=0
 $$ for some vector field $v$ on $\Sigma$.
\end{conjecture}

\begin{remark}
The equivalence of the two statements is 
a consequence of Moser's theorem: 
if the total mass on a surface is preserved, the corresponding evolution of density can be realized 
as a flow of a time-dependent vector field. 
\end{remark}


\subsection{Willmore energy in binormal flows}

\begin{proposition}\label{prop:filament_Willmore}
$i)$ Conjecture \ref{conj:Willmore} is true in dimension 1, i.e. for any closed curve in $\mathbb R^3$:
$$
\mathcal W(\gamma_t)=const.
$$

$ii)$ The Willmore energy is not necessarily  invariant for membranes, i.e. in dimension $n\geq2$.
Namely, for the binormal evolution of the sphere products 
$F:\Sigma=\mathbb S^m(a)\times\mathbb S^l(b)\rightarrow\mathbb R^{m+1}\times\mathbb R^{l+1}=\mathbb R^{m+l+2}$
of radiuses $a$ and $b$, the corresponding Willmore energy is not preserved for any initial values of $a$ and $b$:
$$
\mathcal W(F_t)=C_{m,l}\left(\frac{m^2}{a(t)^2}+\frac{l^2}{b(t)^2}\right)\cdot {\rm vol}(\Sigma)\,
$$
for a constant $C_{m,l}$ and  ${\rm vol}(\Sigma)={\rm vol}(F_t(\Sigma)):=a^m b^l$.
\end{proposition} 

\begin{remark}
The Willmore energy of the Clifford torus $F: \mathbb T^2=\mathbb S^1(a)\times \mathbb S^1(b)\to \mathbb R^4$ 
evolves in the binormal flow as follows:
$$
\mathcal W(F_t)=4\pi^2\left(\frac {b(t)}{a(t)}+\frac {a(t)}{b(t)}\right)=4\pi^2\left(\frac ba e^{2t/(ab)}+\frac ab e^{-2t/(ab)}\right).
$$
\end{remark} 

\begin{proof}
For a curve $\gamma\subset\mathbb R^3$, the conservation of the Willmore energy means the time invariance of the integral
$\mathcal{W}(\gamma)=\int_\gamma k^2 \,d s$ or, equivalently, in the arc-length parameterization, 
of the integral 
$\int_\gamma |\gamma''|^2\,d s$. 
The latter invariance follows from this straightforward computation: 
$$
\partial_t\mathcal{W}(\gamma) = 
2\int_\gamma (\partial_t\gamma'',\gamma'')\,d s 
= 
-2\int_\gamma (\partial_t\gamma',\gamma''')\,d s
= 
-2\int_\gamma ((\gamma'\times \gamma'')',\gamma''')\,d s=0. 
$$

In higher dimensions, the evolution of the sphere products is given by the system (\ref{eq:CliffordODE2}). 
It is Hamiltonian on the $(a,b)$-plane with the Hamiltonian function given by $H(a,b):=\ln(a^mb^l)=\ln({\rm vol}(\Sigma))+const$, the logarithm of the volume of the sphere product: ${\rm vol}(\Sigma):=C_{m,l}\,a^mb^l$.
To be invariant, the Willmore energy has to be a function of ${\rm vol}(\Sigma)$ as well. But one obtains
$$
\mathcal W(F_t)=\int_{\Sigma_t} |H|^2 d{\rm vol}_g=\left(\frac{m^2}{a(t)^2}+\frac{l^2}{b(t)^2}\right)\cdot \text{vol}(\Sigma_t)=C_{m,l}\left(\frac{m^2}{a(t)^2}+\frac{l^2}{b(t)^2}\right)\cdot a^mb^l,
$$
where $C_{m,l}$ is a constant depending  on the dimensions $m,\,l$. One observes that factor $\left({m^2}/{a(t)^2}+{l^2}/{b(t)^2}\right)$
in the Willmore energy  cannot be a function of the area $a^mb^l$ (see explicit formulas for $a(t)$ and $b(t)$ in Corollary \ref{cor:2}), hence $\mathcal W(F_t)$ is not preserved.
\end{proof}

\begin{remark}
Furthermore, one can give a simple parametrization to a Clifford torus  and derive explicitly 
its  second fundamental form: $A={\rm diag}(-\frac 1 a {\bf n}_1, -\frac 1 b {\bf n}_2)$.
This example might be particularly useful in order to prove the filament conjecture for membranes for the Gross-Pitaevskii equation, cf. \cite{Jer, JerSm}.
\end{remark}

In  Appendix we will quantify the measure of noninvariance of the Willmore energy in the skew-mean-curvature flows by 
deriving the continuity equation with a source term governing 
the density $\rho=|H|^2$ of the mean curvature:

\begin{theorem}\label{eq:source}
The skew-mean curvature evolution of the membrane $\Sigma$ yields the following continuity equation with a source 
on the``curvature density"  $\rho=|H|^2$:
$$
\partial_t \rho+{\rm div}(\rho\chi)=-2g^{ik}g^{jl}\left( A_{ij},H\right) \left( A_{kl},\,JH\right),
$$
where  $A_{ij}$ are the second fundamental forms in local coordinates, $(g^{ij})$ is the inverse matrix of the induced metric 
$(g_{ij})$, and  $\chi(q)=2g^{ij}\left(\nabla_j^{\perp}\frac{H}{|H|},\frac{JH}{|H|}\right) e_i$ is the torsion vector field (discussed in the next section). 
\end{theorem}

Here and below one assumes the sum over repeated indices.
We give details of the proof and the full system of the equations in  Section \ref{subsec:Continuity equation}.


\medskip

\section{Torsion forms and torsion vector fields for membranes}\label{sec:torsion}

While torsion is a classical intrinsic notion in Riemannian geometry, for codimension 2 membranes one can
introduce a natural torsion of their embedding into the ambient Euclidean space, similar to that of curves in $\mathbb R^3$.
Namely, for such curves, according to the Frenet--Serret formulas,
the curvature vector $\kappa \,\bf n$ is described by its magnitude and the angle of rotation in the normal plane as a function of the curve parametrization.
Similarly, for codimension 2 membranes, one can define the mean curvature vector $H$ in the normal plane, while
its ``angle of rotation" leads to the following definition of the torsion connection form in the (normal) $S^1$-bundle over the membrane.

For an immersed submanifold $F:\Sigma^n\rightarrow\mathbb R^{n+2}$ consider the principal $S^1$-bundle $N$ 
of unit normal vectors over $\Sigma$: 
\begin{figure}[H]
\centerline{
\xymatrix{
S^1\ar[r] &N \ar[d]^{\pi}\\
&\Sigma^n
}}
\caption{The $S^1$-bundle  of unit normal vectors over $\Sigma$.} \label{fig:diag2}
\end{figure}
\noindent
Let $H$ be the field of mean curvature vectors over $\Sigma$, and we assume that $|H|\neq 0$ everywhere
(otherwise we pass to the open part $\Sigma ^*\subset \Sigma$ where $H$ is nonvanishing, as our consideration is local).
Then the normalized vectors $h:=\frac{H}{|H|}$ define a smooth section of the $S^1$-bundle  $N$.

\begin{definition}
The (normal) reference connection $A_0$ on $N$  is defined by setting the tangent space to 
$\Sigma$ to be its horizontal space.  Then any other connection $A$ in $N$ can be expressed 
as a connection (and hence a 1-form) on the base by comparing it to the connection $A_0$.
\end{definition}

Recall that  for any principal $G$-bundle $\pi: P\rightarrow B$, all its connections $A\in\Omega^1(P,\mathfrak g)$ form 
an affine space. Upon fixing a reference connection  $A_0\in\Omega^1(P,\mathfrak g)$,  any other connection 
$A\in\Omega^1(P,\mathfrak g)$ can be expressed via the difference $A-A_0\in\Omega^1(B,\mathfrak g)$, which is a $\mathfrak g$-valued 1-form on the base. In the case of $S^1$-bundle, this difference becomes a real-valued 1-form.

\begin{definition}\label{def:torsion}
The (mean) curvature connection $A$ on the principal $S^1$-bundle $N$ is defined by declaring 
the section   $h:=\frac{H}{|H|}:\Sigma\rightarrow \mathbb R^{n+2}$ be its horizontal section.
The generalized {\it torsion form} of the submanifold $F:\Sigma\rightarrow\mathbb R^{n+2}$ is the 1-form $\tau=A-A_0\in\Omega^1(\Sigma,\mathbb R)$, where $A_0\in\Omega^1(N,\mathbb R)$ is the normal connection  described above.

The {\it torsion vector field}  $\chi:=2\tau^\sharp$ is defined as metric dual to the  torsion form, i.e.
for any  vector  $v\in T_q\Sigma$  one sets  $(\chi,v)=2\tau(v)$ at  any point $q\in \Sigma$.
\end{definition}

\begin{proposition}\label{prop:dtau}
The 2-form $-d\tau$ is equal to the normal curvature of the submanifold $F:\Sigma\rightarrow\mathbb R^{n+2}$.
\end{proposition}

\begin{proof}
The exterior covariant derivative in $N$ is just the exterior derivative, since $S^1$ is abelian. Hence the curvature of the connection $A$ is $\Omega=dA= d\tau+dA_0$. Furthermore, $A$ is a flat connection, since $\frac{H}{|H|}$ is a global section.
We obtain that  $d\tau=-dA_0$, which means $-d\tau$ coincides with the  normal curvature, since $A_0$ is  induced by the normal connection.
\end{proof}

\begin{remark}
Proposition \ref{prop:dtau} emphasizes an important difference between the higher-dimensional and 1D cases: in higher dimensions the torsion form  $\tau$ is not exact in general, which partially explains the absence of the Hasimoto transform: one cannot introduce the ``phase" of the would-be wave function, i.e. the ``angle of rotation" of the mean curvature vector $H$, as it depends not only on a point $q\in \Sigma$, but also on a path along the membrane 
$\Sigma$ from a reference point $q_0$ to $q$. 
Furthermore, as we mentioned before, in higher dimensions the density $\rho:=|H|^2$ is not transported by the torsion 
vector field $\chi$ related to $\tau$, but satisfies the continuity equation with a source term (Theorem \ref{eq:source}).
\end{remark}

\medskip

\section{Appendix: Generalized Da Rios equations}\label{sec:DaRios}

The evolution of the codimension 2 membranes according to the binormal flow
satisfies a system of equations on its mean curvature vector $H$ and generalized torsion form $\tau$.
Here we derive those  {\it generalized Da Rios-type equations.} 
Due to their analogy with the compressible fluid equations, we will call the equation on 
the mean curvature  $H$  the continuity equation, while the evolution of the torsion form $\tau$
is the momentum equation. Some computations in this section can be found, e.g., in \cite{KuSc, SongSun}, and are included here to make the derivation of the Da Rios-type equations \eqref{eq:conti_eq}-\eqref{eq:moment_eq} self-contained.

\medskip

\subsection{Gradient of the Willmore energy}\label{subsec:willmore}

We start by deriving the gradient of the Willmore energy in any dimension, which could be of independent interest. 
For this we generalize the derivation of the Willmore gradient done in \cite{KuSc} 
for 2-dimensional, compact immersed surfaces in $\mathbb R^{d}$ to the case of   compact immersed submanifolds
of any dimension.

More specifically, consider an immersed submanifold $F:\Sigma^n\rightarrow\mathbb R^{n+k}$. Recall that the {\it Willmore energy} is defined as
\begin{equation}\label{willmore}
{\mathcal W}(F)=\int_{\Sigma}| H|^2\;d\text{vol}_{g},
\end{equation}
where $g=F^*g_e$ denotes the pull back metric of the Euclidean metric $g_e$ on $\mathbb R^{n+k}$ and $H$ is the corresponding mean curvature vector field.

In local coordinates $ (x_1,...,x_n)$ on the manifold $\Sigma$ the pull-back metric $g$ on $\Sigma$ is
$$
g_{ij}=\left( \partial_i F,\partial_j F \right),
$$
and the corresponding volume element is the $n$-form $d\text{vol}=\sqrt{\text{det}g_{ij}}\;dx_1\wedge\dots\wedge dx_n$.

We have the following splitting of the pull-back bundle $F^*T\mathbb R^{n+k}=\bigcup_{q\in \Sigma}T_{F(q)}\mathbb R^{n+k}$:
$$
T_{F(q)}\mathbb R^{n+k}=DF|_{q}(T_q\Sigma)\oplus N_p\Sigma,
$$
where $DF$ is the tangent map of $F$.
The second fundamental form $A_{ij}=(\partial_i \partial_j F)^{\perp}$ is the projection of the second derivatives of $F$ to the normal bundle $ N_p\Sigma$. Then the mean curvature vector at any point is $ H=g^{ij}A_{ij}$, where $(g^{ij})$ is the inverse matrix of the induced metric $(g_{ij})$.

Now we give the formula of the normal gradient of the Willmore energy.

\begin{theorem}\label{prop:gradient}
The normal part of the gradient of the Willmore energy is 
\begin{equation}\label{gradient}
\frac 12\nabla^{\perp} {\mathcal W}=\Delta^{\perp} H+g^{ik}g^{jl}\left( A_{ij},H\right) A_{kl}-\frac 12 |H|^2H,
\end{equation}
where $\Delta^{\perp}=g^{ij}\nabla^{\perp}_i\nabla^{\perp}_j$ denotes the Laplacian in the normal bundle, and $\nabla^{\perp}_i=\nabla^{\perp}_{\frac{\partial}{\partial x_i}}$ is the normal connection.
\end{theorem}

To prove this theorem, we need the following two lemmas, which we include for a self-contained proof.

\begin{lemma}[(cf. \cite{KuSc, SongSun})]\label{lem:vol_derivative}
For a smooth family of immersions $F_t:\Sigma^n\rightarrow\mathbb R^{n+k}$ with  a normal variation $\partial_t F_t|_{t=0}=V$ 
along $F_t$, the time derivative of the volume element is
\begin{equation}\label{eq:vol_derivative} 
\partial_t d{\rm vol}_g=-\left( H,\,V\right)\, d{\rm vol}_g.
\end{equation}
\end{lemma}
\begin{proof}
One has $\partial_t \text{det}(g_{ml}) = (g^{ij}\partial_t g_{ij})\text{det}(g_{ml}),$ and
$$
\partial_t g_{ij} = \left( \partial_t \partial_i F,\partial_j F \right)+\left( \partial_i F, \partial_t\partial_j F \right)
=-\left( \partial_t F, \partial_i \partial_j F \right)-\left( \partial_i \partial_j F, \partial_t F \right)
=-2\left( A_{ij},\,V\right).
$$
From this we obtain
$$
\partial_t \text{det}(g_{ml})=-2g^{ij}\left( A_{ij},\,V\right)\,\text{det}(g_{ml})=-2\left( H,\,V\right)\,\text{det}(g_{ml}).
$$
Therefore,
$$
\partial_t\sqrt{\text{det}(g_{ml})}=\frac {1}{2\sqrt{\text{det}(g_{ml})}}\partial_t \text{det}(g_{ml})=-\left( H,\,V\right)\,\sqrt{\text{det}(g_{ml})},
$$
i.e., $\partial_t d\text{vol}_g=-\left( H,\,V\right)\, d\text{vol}_g.$
\end{proof}

Define the normal  derivative $\partial^{\perp}_t H$ of the mean curvature vector $H$ 
as the projection of the time derivative $\partial_t H$  to the normal bundle to $\Sigma$.

\begin{lemma}[(cf. \cite{KuSc})]\label{lem:H_derivative}
For a smooth family of immersions $F_t:\Sigma^n\rightarrow\mathbb R^{n+k}$ with a normal field $\partial_t F_t|_{t=0}=V$  along $F_t$, the normal time derivative $\partial^{\perp}_t H$ 
of  $H$ is
$$
\partial^{\perp}_t H=\Delta^{\perp}V+g^{im}g^{jl}\left( A_{ij},V\right) A_{ml}.
$$
\end{lemma}
\begin{proof}
Since $A_{ij}=(\partial_i\partial_j F)^{\perp}=\partial_{i}\partial_jF-\Gamma^k_{ij}\partial_k F=\nabla_i\nabla_jF$, one has
$$
\begin{array}{rcl}
\partial^{\perp}_tA_{ij}&=&(\partial_{i}\partial_jV-\Gamma^k_{ij}\partial_k V)^{\perp}=(\nabla_i\nabla_jV)^{\perp}=\nabla^{\perp}_i\nabla^{\perp}_j V+\nabla^{\perp}_i(\left(\partial_jV,\,\partial_mF\right) g^{ml}\partial_lF)\\
                                       &=&\nabla^{\perp}_i\nabla^{\perp}_j V-\left( A_{jm},\,V\right) g^{ml}\nabla^{\perp}_i\partial_lF=\nabla^{\perp}_i\nabla^{\perp}_j V-\left( A_{jm},\,V\right) g^{ml}A_{il}.
\end{array}
$$
For $H=g^{ij}A_{ij}$ we obtain
$$
\begin{array}{rcl}
\partial^{\perp}_tH&=&g^{ij}(\partial^{\perp}_tA_{ij})+(\partial^{\perp}_tg^{ij})A_{ij}\\
                                &=&g^{ij}(\nabla^{\perp}_i\nabla^{\perp}_j V-\left( A_{jm},\,V\right) g^{ml}A_{il})+2g^{im}g^{jl}\left( A_{ml},\,V\right) A_{ij}\\
                                       &=&\Delta^{\perp}V+g^{im}g^{jl}\left( A_{ij},V\right) A_{ml}.
\end{array}
$$
\end{proof}

Now we can complete the proof of the theorem.

\begin{proof}[Proof of Theorem \ref{prop:gradient}]
Consider a smooth family of immersions $F_t:\Sigma^n\rightarrow\mathbb R^{n+k}$ with the field
$\partial_t F_t|_{t=0}=V$ normal along $F_t$. Then the time derivative of the Willmore energy is
$$
\begin{array}{rcl}
\partial_t {\mathcal W}(F_t)&=&\int_{\Sigma}\partial_t(|H|^2\,d\text{vol}_g)= 2\int_{\Sigma}\left( \partial^{\perp}_tH,\,H\right)\,d\text{vol}_g+\int_{\Sigma}|H|^2\partial_t d\text{vol}_g\\
&=&2\int_{\Sigma}\left( \Delta^{\perp}V,\,H\right)+g^{im}g^{jl}\left( A_{ij},V\right) \left( A_{ml},\,H\right)-\frac 12 |H|^2\left( H,\,V\right)\,d\text{vol}_g\\
&=& 2\int_{\Sigma}\left( \Delta^{\perp} H+g^{im}g^{jl}\left( A_{ij},H\right) A_{ml}-\frac 12 |H|^2H,\,V\right) \,d\text{vol}_g.
\end{array}
$$
Therefore, the normal gradient of the Willmore energy is 
$$\frac 12\nabla^{\perp} {\mathcal W}=\Delta^{\perp} H+g^{im}g^{jl}\left( A_{ij},H\right) A_{ml}-\frac 12 |H|^2H.$$
\end{proof}


\subsection{Evolution of the Willmore energy in the skew-mean-curvature flows}\label{subsec:time_deri}

Consider now a smooth family of immersions $F_t:\Sigma^n\rightarrow\mathbb R^{n+2}$  evolved by the skew-mean-curvature flow: $\partial_t F_t|_{t=0}=-JH$, where $J$ is the operator of rotation by $\pi/2$ in the positive direction in every normal space to $\Sigma$.

\begin{proposition}
The Willmore energy of $\Sigma$ changes in time in the skew-mean-curvature  flow as follows:
$$
\partial_t {\mathcal W}(F_t)=-2\int_{\Sigma}\left( A^l_{i},H\right) \left( A^i_{l},JH\right) \,d{\rm vol}_g.
$$
\end{proposition}

\begin{proof}
Employing the gradient formula for the Willmore energy 
we obtain 
$$
\begin{array}{rcl}
\partial_t {\mathcal W}(F_t)&=& -2\int_{\Sigma}\left(\nabla {\mathcal W},\,JH\right)\,d\text{vol}_g
= -2\int_{\Sigma}\left(\nabla^{\perp} {\mathcal W},\,JH\right)\,d\text{vol}_g\\
&=&-2\int_{\Sigma}\left( \Delta^{\perp} H+g^{im}g^{jl}\left( A_{ij},H\right) A_{ml}-\frac 12 |H|^2H,\,JH\right) \,d\text{vol}_g.
\end{array}
$$
The third term $\left( \frac 12 |H|^2H,\,JH\right)$ is pointwise zero on $\Sigma$, since $H\perp JH$. 

For the first term we need the following lemma.
\begin{lemma}[(see \cite{SongSun})]\label{lem:commute}
$\nabla^{\perp}J=J\nabla^{\perp}.$
\end{lemma}
\begin{proof}
Let us prove that $\nabla^{\perp}JV=J\nabla^{\perp}V$ for an arbitrary unit normal vector field $V$. Note that
$\{V,U=JV\}$ form a local orthonormal frame. Hence for any tangent vector field $X$, we have
$$
J\nabla^{\perp}_X V=J(\partial_XV)^{\perp}=J(\left(\partial_X V, V\right) V +\left(\partial_X V,U\right) U)=J\left(\partial_X V,U\right) U,
$$
and
$$
\nabla^{\perp}_X(JV) =\nabla^{\perp}_X(U) = (\partial_XU)^{\perp}=\left(\partial_XU, V\right) V +\left(\partial_XU,U\right) U=\left(\partial_XU, V\right) V=J\left(\partial_X V,U\right) U.
$$
Therefore $J\nabla^{\perp}_X V=\nabla^{\perp}_X(JV)$.
\end{proof}

To complete the proof of the proposition, we integrate by parts 
the first term of $\partial_t {\mathcal W}(F_t)$:
$$
-\int_{\Sigma}\left( \Delta^{\perp} H,\,JH\right) \,d\text{vol}_g=\int_{\Sigma}\left( \nabla^{\perp} H,\,\nabla^{\perp} JH\right) \,d\text{vol}_g,
$$
and by Lemma \ref{lem:commute}, $\left( \nabla^{\perp} H,\,\nabla^{\perp} JH\right)=\left( \nabla^{\perp} H,\,J\nabla^{\perp} H\right)=0$  pointwise on $\Sigma$, i.e., $-\int_{\Sigma}\left( \Delta^{\perp} H,\,JH\right) \,d\text{vol}_g$ vanishes on $\Sigma$.
Thus we conclude that 
$$
\partial_t {\mathcal W}(F_t)=-2\int_{\Sigma}g^{im}g^{jl}\left( A_{ij},H\right) \left( A_{ml},JH\right) \,d\text{vol}_g=-2\int_{\Sigma}\left( A^l_{i},H\right) \left( A^i_{l},JH\right) \,d\text{vol}_g.
$$
\end{proof}

\begin{corollary}
The Willmore energy of a closed submanifold $\Sigma$
is invariant under the skew-mean-curvature flow, if and only if 
$$
\int_{\Sigma}\left( A^l_{i},H\right) \left( A^i_{l},JH\right) \,d{\rm vol}_g=0
$$
for all times $t$.
\end{corollary}

\begin{remark}
For a 1-dimensional $\Sigma$, i.e for a vortex filament $\gamma$, the second fundamental form reduces to the mean curvature of the curve:  $A=H=\kappa$. Hence $\left( A,H\right) \left( A,JH\right)=0$ at every point on $\Sigma$. However, in higher dimensions the pointwise identity $\left( A^l_{i},H\right) \left( A^i_{l},JH\right)=0$  does not necessarily hold on the membrane $\Sigma$, and hence the Willmore energy might not conserve.
\end{remark}

\begin{remark}
For a Clifford torus the computation of  $\int_{\mathbb T^2}\left( A_{ij},H\right) \left( A_{ij},JH\right) \,d\text{vol}_g$ 
is straightforward: since
$$
\left( A_{ij},H\right)\left( A_{ij},JH\right)=-\frac {1}{a^2}\frac{1}{ab}+\frac {1}{b^2}\frac{1}{ab}=-\frac{1}{a^3b}+\frac{1}{ab^3},
$$
we obtain
$$
\int_{\mathbb T^2} \left( A_{ij},H\right)\left( A_{ij},JH\right) \,d{\rm vol}_g=\int_0^{2\pi}\int_0^{2\pi}\left(-\frac{1}{a^3b}+\frac{1}{ab^3}\right)ab\,d\theta d\phi=4\pi^2\left(\frac {1}{b^2}-\frac{1}{a^2}\right).
$$

So if $a\neq b$ at time $t$, we have $\int_\Sigma \left( A_{ij},H\right)\left( A_{ij},JH\right) \,d{\rm vol}_g\neq 0$. Furthermore,
the tori with equal radiuses $a=b$ do not form an invariant set, since under the skew-mean-curvature flow one radius of  
the Clifford torus is increasing, while the other is decreasing. 
Thus the Willmore energy is not invariant under the skew-mean-curvature flow for any initial values $a$ and $b$. 
\end{remark}
\medskip

\subsection{The continuity equation and generalized Da Rios equations}\label{subsec:Continuity equation}

Let $F_t:\Sigma^n\rightarrow\mathbb R^{n+2}$ be  a codimension 2 vortex membrane moving by the skew-mean-curvature flow. Let $(x_1,x_2,\cdots,x_n)$ be local coordinates on $\Sigma^n$, and $\{e_1,e_2,\cdots,e_n\}$  be the corresponding local frame in the tangent space.

According to Lemma \ref{lem:H_derivative},  the time derivative of the square of the mean curvature is 
\begin{equation}\label{eq:density}
\partial_t|H|^2=-2\left( \Delta^{\perp} H+g^{ik}g^{jl}\left( A_{ij},H\right) A_{kl},\,JH\right).
\end{equation}
It turns out that the first term in the right-hand side can be expressed as the divergence of a certain vector field.

\begin{lemma}\label{lem:divergence}
For a  vector field
$\sigma:=\left( g^{ik}\nabla^{\perp}_kH,\,JH\right)e_i$ on the submanifold $\Sigma$, its divergence is
$$
{\rm div}_{\Sigma}\,\sigma=\left( \Delta^{\perp}H,JH\right).
$$
\end{lemma}

\begin{proof}
Recall that  for an arbitrary vector field $X=X^ie_i$ on $\Sigma$ its divergence is as follows:
\begin{equation}\label{eq:divergence}
{\rm div}_{\Sigma}X= \text{tr}(\nabla X)=g^{ij}\nabla_iX_j=\nabla_iX^i.
\end{equation}
Furthermore, define $T^i=g^{ik}\nabla^{\perp}_kH$. Then 
we have $\nabla^{\perp}_iT^i=\nabla^{\perp}_ig^{ij}\nabla^{\perp}_jH=\Delta^{\perp}H$. This implies that
$$
\begin{array}{rcl}
{\rm div}_{\Sigma} \sigma&=&\nabla_i\left( T^i,\,JH\right)=\left( \nabla_iT^i,\,JH\right)+\left( T^i,\,\nabla_iJH\right)\\
&=&\left( \Delta^{\perp}H,JH\right)+\left( g^{ij}\nabla^{\perp}_jH,\nabla_iJH\right)=\left(\Delta^{\perp}H,JH\right)\,.
\end{array}
$$
\end{proof}

Recall that the torsion form $\tau=\tau_idx^i$ has components 
$\tau_i=\left(\nabla_i^{\perp}\frac{H}{|H|},\frac{JH}{|H|}\right)$ for the mean curvature vectors $H$ on a membrane
$\Sigma$ with the second fundamental form $A_{ij}$. We can now describe explicitly the  torsion vector field introduced in Definition \ref{def:torsion}.

\begin{proposition}\label{prop:chi} Given a local frame  $\{e_1,e_2,\cdots,e_n\}$ in the tangent space of $\Sigma$ 
 the  torsion vector field in the corresponding local coordinates is 
$$
\chi=2g^{ij}\left(\nabla_j^{\perp}\frac{H}{|H|},\frac{JH}{|H|}\right) e_i\,,
$$ 
where $(g^{ij})$ is the inverse matrix of the metric 
$(g_{ij})$ induced on $\Sigma$ from $ \mathbb R^{n+2}$.
\end{proposition}

\begin{proof}
For a vector $v\in T_q\Sigma$ and the section $h:=\frac{H}{|H|}$, the tangent map $Dh:T_q\Sigma\rightarrow T_{h(q)}N$ maps $v$ to the vector $Dh(v)$ in the tangent space of the section $h$. Then the normal component of $Dh(v)$ is equal to $\tau(v)=(A-A_0)(v)$.

Denote by $\nabla_i h$ the vector $Dh(e_i)$ in the tangent space of the smooth section $h$ over $\Sigma$ at any point, 
then $\left(\nabla_ih,Jh\right)Jh$ is its normal component.
For $v=v^ie_i$  the normal component of $Dh(v)$ is $v^i\left(\nabla_ih,Jh\right)Jh$, hence the torsion form is
$$
\tau(v)=v^i\left(\nabla_ih,Jh\right)=\frac12(\chi,v),
$$
where $\chi=2g^{ij}\left(\nabla_j^{\perp}h,Jh\right) e_i$ is the torsion vector field.
\end{proof}

Recall that for the normalized mean curvature  $h:={H}/{|H|}$ the orthonormal frame $\{h, Jh\}$ is a basis of the normal bundle
$N\Sigma$ 
to the membrane $\Sigma$. The  torsion form $\tau=\tau_idx^i$  for $\tau_i=\left(\nabla_i^{\perp}h,Jh\right)$ measures 
how much this frame rotates when one moves along the tangent vector $e_i$ on the surface $\Sigma$. 
Finally, the equations for torsion and curvature (or, rather, curvature density) form  the following pair of equations,  generalizing the Da Rios system \eqref{eq:DaRios}.

\begin{theorem}[{\bf =\eqref{eq:source}$'$}]\label{thm: conti_eq}
The skew-mean-curvature evolution of the membrane $\Sigma$ implies the following continuity equation with a source 
for the curvature density  $\rho=|H|^2$
\begin{equation}\label{eq:conti_eq}
\partial_t \rho+{\rm div}(\rho\chi)=-2g^{ik}g^{jl}\left( A_{ij},H\right) \left( A_{kl},\,JH\right)
\end{equation}
 and the momentum equation for the torsion $\tau=\tau_idx^i$
\begin{equation}\label{eq:moment_eq}
\begin{array}{rcl}
\partial_t{\tau_i}+\nabla_i|\tau|^2-\nabla_i\frac{\Delta|H|}{|H|}&=&-\nabla_i\frac{g^{mk}g^{jl}( A_{mj},JH)( A_{kl},\,JH)}{|H|^2}\\\\
&+&\frac{g^{kl}}{|H|^2}\left(\left(A_{ik},H\right)\left(\nabla_lJH,JH\right)-\left(A_{il},JH\right)\left(\nabla_kJH,H \right)\right).
\end{array}
\end{equation}
\end{theorem}

\begin{corollary}
The continuity equation \eqref{eq:conti_eq} can be rewritten in the form
\begin{equation}\label{eq:conti_eq2}
\partial^{\perp}_t H+2g^{ij}\tau_i\nabla^{\perp}_jH+(\nabla^i\tau_i)H=-g^{ik}g^{jl}\left( A_{kl},\,JH\right)  A_{ij}.
\end{equation}
\end{corollary}

\begin{proof}
Equation (\ref{eq:conti_eq})  follows from (\ref{eq:density}), Lemma \ref{lem:divergence}, and Proposition \ref{prop:chi}, since 
$\chi=2g^{ij}\left(\nabla_j^{\perp}\frac{H}{|H|},\frac{JH}{|H|}\right) e_i=2g^{ij}\frac{\left( \nabla_j^{\perp}H,JH\right)}{|H|^2}e_i$.
Plugging $\rho=|H|^2$ into (\ref{eq:conti_eq}), we get
$$
2(H,\partial_t H)+|H|^2\text{div}\chi+2(H,\nabla_{\chi}H)=-2g^{ik}g^{jl}\left( A_{ij},H\right) \left( A_{kl},\,JH\right).
$$
Then plugging in $\chi=2g^{ij}\tau_je_i$ we obtain
$$
2(H,\partial_t H)+4(H,g^{ij}\tau_i\nabla_jH)+2|H|^2\nabla^i\tau_i=-2g^{ik}g^{jl}\left( A_{ij},H\right) \left( A_{kl},\,JH\right),
$$
i.e.
$$
\partial^{\perp}_t H+2g^{ij}\tau_i\nabla^{\perp}_jH+(\nabla^i\tau_i)H=-g^{ik}g^{jl}\left( A_{kl},\,JH\right)  A_{ij}.
$$

The momentum equation is obtained by using $\partial_t\left(\nabla_i\frac{H}{|H|}\right)=\nabla_i\left(\partial_t\frac{H}{|H|}\right)$ via a direct but tedious computation comparing the corresponding coefficients.
\end{proof}

\begin{remark}
The equations \eqref{eq:conti_eq}-\eqref{eq:moment_eq} are  analogues of the equations of barotropic-type fluids \eqref{eq:barotropic2} related to vortex membranes in higher dimensions, as well as the natural extensions of  the Da Rios system (\ref{eq:DaRios}). Their  more complicated form in higher dimensions is related to the fact that the metric
induced on a membrane changes during the binormal evolution, while in the 1D case the induced metric on a vortex filament (e.g. arc-length parametrization) remains intact due to inextensibility of the curve.
\end{remark}


%

\end{document}